\documentclass[11pt]{article}

\pdfoutput=1

\usepackage{a4}
\usepackage[top=30truemm,bottom=30truemm]{geometry}

\usepackage[running]{lineno}

\usepackage{authblk}
\usepackage{amsmath}
\usepackage{amssymb}
\usepackage{amsfonts}
\usepackage{amsthm}
\usepackage{cases}
\usepackage{url}
\usepackage{graphicx}

\usepackage{url}
\usepackage{multirow}
\usepackage{multicol}
\usepackage{cases}
\usepackage[final,mode=multiuser]{fixme}
\fxuseenvlayout{colorsig}
\fxusetargetlayout{color}
\FXRegisterAuthor{sm}{asm}{SM}
\FXRegisterAuthor{yn}{ayn}{YN}
\FXRegisterAuthor{si}{asi}{SI}
\FXRegisterAuthor{hb}{ahb}{HB}
\FXRegisterAuthor{tk}{atk}{TK}
\fxsetface{env}{}

\newtheorem{Theorem}{Theorem}
\newtheorem{Lemma}{Lemma}
\newtheorem{Corollary}{Corollary}

\theoremstyle{definition}
\newtheorem{Definition}{Definition}
\newtheorem{Problem}{Problem}
\newtheorem{MyRemark}{Remark}
\newtheorem{MyExample}{Example}

\newcommand{\lcp}{\mathsf{lcp}}

\newcommand{\hd}{\mathsf{d_H}}

\newcommand{\argmax}{\mathop{\rm arg~max}\limits}

\newcommand{\NLZ}{\mathsf{LZN}}
\newcommand{\zn}{\mathsf{zn}}
\newcommand{\SLZ}{\mathsf{LZS}}
\newcommand{\zs}{\mathsf{zs}}

\newcommand{\CN}{\mathsf{CN}}
\newcommand{\cn}{\mathsf{cn}}

\newcommand{\avl}{\mathsf{avl}}

\newcommand{\poly}{\mathrm{poly}}

\title{Compressed Communication Complexity of \\ Hamming Distance}

\author{Shiori~Mitsuya$^{1}$}
\author{Yuto~Nakashima$^{1}$}
\author{Shunsuke~Inenaga$^{1,2}$}
\author{Hideo~Bannai$^3$}
\author{Masayuki~Takeda$^{1}$}

\affil{
  \normalsize{
  \textit{$^1$Department of Informatics, Kyushu University, Fukuoka, Japan}\\
  \texttt{\{mitsuya.shiori, yuto.nakashima, inenaga, takeda\}@inf.kyushu-u.ac.jp}\\
  \textit{$^2$PRESTO, Japan Science and Technology Agency, Kawaguchi, Japan}\\
  \textit{$^3$M\&D Data Science Center, Tokyo Medical and Dental University, Tokyo, Japan}\\
  \texttt{hdbn.dsc@tmd.ac.jp}
  }
}

\date{}

\begin{document}
\maketitle

\begin{abstract}
  We consider the communication complexity of
  the \emph{Hamming distance} of two strings.
  Bille et al. [SPIRE 2018] considered
  the communication complexity of the longest common prefix (LCP) problem
  in the setting where the two parties
  have their strings in a compressed form,
  i.e., represented by the Lempel-Ziv 77 factorization (LZ77)
  with/without self-references.
  We present a randomized public-coin protocol
  for a joint computation of the Hamming distance
  of two strings represented by LZ77 without self-references.
  While our scheme is heavily based on Bille et al.'s LCP protocol,
  our complexity analysis is original which uses
  Crochemore's C-factorization and Rytter's AVL-grammar.
  As a byproduct, we also show that
  LZ77 with/without self-references are not monotonic
  in the sense that their sizes can increase by a factor of
  4/3 when a prefix of the string is removed.
\end{abstract}

\section{Introduction}

Communication complexity, first introduced by Yao~\cite{Yao79},
is a well studied sub-field of complexity theory
which aims at quantifying the total amount of communication (bits)
between the multiple parties who separately hold partial inputs of a function $f$.
The goal of the $k~(\geq 2)$ parties is to jointly compute the value of $f(X_1, \ldots, X_k)$,
where $X_i$ denotes the partial input that the $i$th party holds.
Communication complexity studies lower bounds and upper bounds
for the communication cost of a joint computation of a function $f$.
Due to the rapidly increasing amount of distributed computing tasks,
communication complexity has gained its importance
in the recent highly digitalized society.
This paper deals with the most basic and common setting where
the two parties, Alice and Bob,
separately hold partial inputs $A$ and $B$
and they perform a joint computation of $f(A, B)$ for a function $f$
following a specified protocol.

We pay our attention to communication complexity of string problems
where the inputs $A$ and $B$ are strings over an alphabet $\Sigma$.
Communication complexity of string problems has played
a critical role in the space lower bound analysis of several streaming processing
problems including
Hamming/edit/swap distances~\cite{CliffordJPS13},
pattern matching with k-mismatches~\cite{RadoszewskiS20},
parameterized pattern matching~\cite{JalseniusPS13},
dictionary matching~\cite{GawrychowskiS19},
and quasi-periodicity~\cite{GawrychowskiRS19}.

Bille et al.~\cite{Bille18} were the first to consider
the communication complexity of the
longest common prefix (LCP) problem in the setting where the two parties
have their strings in a compressed form,
i.e., represented by the \emph{Lempel-Ziv 77 factorization} (\emph{LZ77})~\cite{LZ77}
with/without self-references.
Bille et al.~\cite{Bille18} proposed
a randomized public-coin protocol for the LCP problem with
$O(\log z_\ell)$ communication rounds and $O(\log \ell)$ total bits of communication,
where $\ell$ denotes the length of the LCP of the two strings $A$ and $B$
and $z_{\ell}$ denotes the size of the \emph{non self-referencing} LZ77 factorization
of the LCP $A[1..\ell]$.
In addition, Bille et al.~\cite{Bille18} showed a randomized public-coin protocol
for the LCP problem with
\begin{enumerate}
      \item[(i)] $O(\log z_\ell' + \log \log \ell)$ communication rounds
            and $O(\log \ell)$ total bits of communication, or
      \item[(ii)] $O(\log z_\ell')$ communication rounds
            and $O(\log \ell + \log \log \log n)$ total bits of communication,
\end{enumerate}
where $z_\ell'$ denotes the size of the \emph{self-referencing} LZ77 factorization of
the LCP $A[1..\ell]$ and $n = |A|$.

In this paper, we consider the communication complexity of
the \emph{Hamming distance} of two strings of equal length,
which are represented in a compressed form.
We present a randomized public-coin protocol
            for a joint computation of the Hamming distance
            of two strings represented by \emph{non self-referencing} LZ77,
            with $O(d \log z)$ communication rounds and $O(d \log \ell_{\max})$ total bits of communication,
            where $d$ is the Hamming distance between $A$ and $B$,
            $z$ is the size of the LZ77 factorization of string $A$,
            and $\ell_{\max}$ is the largest gap between two adjacent mismatching positions
            between $A$ and $B$\footnote{If the first/last characters of $A$ and $B$ are equal, then we can add terminal symbols as $\#A\$$ and $\$B\#$ and subtract 2 from the computed distance.}.
            While our scheme is heavily based on Bille et al.'s LCP protocol,
            our complexity analysis is original which uses
            Crochemore's C-factorization~\cite{Crochemore84} and Rytter's AVL-grammar~\cite{Rytter03}.

Further, as a byproduct of our result for the Hamming distance problem,
we also show that LZ77 with/without self-references are \emph{non-monotonic}.
For a compression algorithm $\mathsf{A}$
let $\mathsf{A}(S)$ denote the size of the compressed representation of string $S$
by $\mathsf{A}$.
We say that compression algorithm $\mathsf{A}$ is \emph{monotonic}
if $\mathsf{A}(S[1..j]) \leq \mathsf{A}(S)$
for any $1 \leq j < |S|$
and $\mathsf{A}(S[i..|S|]) \leq \mathsf{A}(S)$
for any $1 < i \leq |S|$,
and we say it is \emph{non-monotonic} otherwise.
It is clear that LZ77 with/without self-references
satisfy the first property,
however, to our knowledge the second property has not been studied
for the LZ77 factorizations.
We prove that LZ77 with/without self-references is \emph{non-monotonic}
by giving a family of strings such that
removing each prefix of length from $1$ to $\sqrt{n}$
increases the number of factors in the LZ77 factorization by a factor of 4/3,
where $n$ denotes the string length.
We also show that in the worst-case
the number of factors in the non self-referencing LZ77 factorization
of any suffix of any string $S$ of length $n$
can be larger than that of $S$ by at most a factor of $O(\log n)$.

Monotonicity of compression algorithms
and string repetitive measures has gained recent attention.
Lagarde and Perifel~\cite{LagardeP18}
showed that \emph{Lempel-Ziv 78 compression}~\cite{LZ78} is non-monotonic
by showing that removing the first character of a string
can increase the size of the compression by a factor of $\Omega(\log n)$.
The recently proposed repetitive measure called the
\emph{substring complexity} $\delta$ is known to be monotonic~\cite{KociumakaNP20}.
Kociumaka et al.~\cite{KociumakaNP20} posed an open question
whether the smallest bidirectional macro scheme size $b$~\cite{StorerS82}
or the smallest string attractor size $\gamma$~\cite{KempaP18} is monotonic.
It was then answered by Mantaci et al.~\cite{MantaciRRRS21} that $\gamma$ is non-monotonic.

\section{Preliminaries}
\label{sec:preliminaries}

\subsection{Strings}

Let $\Sigma$ be an {\em alphabet} of size $\sigma$.
An element of $\Sigma^*$ is called a {\em string}.
The length of a string $S$ is denoted by $|S|$.
The empty string $\varepsilon$ is the string of length 0,
namely, $|\varepsilon| = 0$.
The $i$-th character of a string $S$ is denoted by
$S[i]$ for $1 \leq i \leq |S|$,
and the \emph{substring} of a string $S$ that begins at position $i$ and
ends at position $j$ is denoted by $S[i..j]$ for $1 \leq i \leq j \leq |S|$.
For convenience, let $S[i..j] = \varepsilon$ if $j < i$.
Substrings $S[1..j]$ and $S[i..|S|]$
are respectively called a \emph{prefix} and a \emph{suffix} of $S$.
For simplicity, let $S[..j]$
denote the prefix of $S$ ending at position $j$
and $S[i..]$ the suffix $S[i..|S|]$ of $S$ beginning at position $i$.
\hbnote*{fix}{
  A suffix $S[j..]$ with $j > 1$ is called a \emph{proper suffix} of $S$.
}

For string $X$ and $Y$,
let $\lcp(X, Y)$ denote the length of the
\emph{longest common prefix} (LCP) of strings $X,Y$,
namely, $\lcp(X, Y) = \max(\{\ell \mid X[..\ell] = Y[..\ell], 1 \leq \ell \leq \min\{|X|, |Y|\}\} \cup \{0\})$.
The \emph{Hamming distance} $\hd(X,Y)$ of two strings $X, Y$ of equal length
is the number of positions where the underlying characters
differ between $X$ and $Y$,
namely, $\hd(X,Y) = |\{i \mid X[i] \neq Y[i], 1 \leq i \leq |X|\}|$.

\subsection{Lempel-Ziv 77 factorizations}

Of many versions of Lempel-Ziv 77 factorization~\cite{LZ77}
which divide a given string in a greedy left-to-right manner,
the main tool we use is the non self-referencing LZ77,
which is formally defined as follows:

\begin{Definition}[Non self-referencing LZ77 factorization]
  The \emph{non self-referencing LZ77 factorization} of string $S$,
  denoted $\NLZ(S)$,
  is a factorization $S = f_1 \cdots f_\mathsf{zn}$ that satisfies
  the following:
  Let $u_i$ denote the beginning position of each factor $f_i$ in
  the factorization $f_1 \cdots f_\mathsf{zn}$,
  that is, $u_i = |f_1 \cdots f_{i-1}|+1$.
  (1) If $i > 1$ and $\max_{1 \leq j < u_i}\{\lcp(S[u_i..], S[j..u_i-1])\} \geq 1$,
  then for any position $s_i \in \argmax_{1 \leq j < u_i} \lcp(S[u_i..], S[j..u_i-1])$ in $S$, let $p_i = \lcp(S[u_i..], S[s_i..u_i-1])$.
  (2) Otherwise, let $p_i = 0$.
  Then, $f_i = S[s_i..u_i+p_i]$ for each $1 \leq i \leq \mathsf{zn}$.
\end{Definition}

Intuitively, each factor $f_i$ in $\NLZ(S)$ is either a fresh letter,
or the shortest prefix of $f_i \cdots f_\zn$
that does not have a previous occurrence in $f_1 \ldots f_{i-1}$.
This means that self-referencing is \emph{not} allowed in $\NLZ(S)$, namely,
no previous occurrences $S[s_i..s_i+p_i]$ of each factor $f_i$
can overlap with itself.

The \emph{size} $\zn(S)$ of $\NLZ(S)$ is the number $\zn$
of factors in $\NLZ(S)$.

We encode each factor $f_i$ by a triple $(s_i, p_i, \alpha_i) \in ([1..n] \times [1..n] \times \Sigma)$,
where $s_i$ is the left-most previous occurrence of $f_i$,
$p_i$ is the length of $f_i$, and $\alpha_i$ is the last character of $f_i$.

\begin{MyExample}
  For $S = \mathtt{abaababaabaabaabaabaabb}$,
  $\SLZ(S) = \mathtt{a \mid b \mid aa \mid bab \mid aabaa \mid baabaab \mid aabb \mid}$
  and it can be represented as
  $(0, 0, \mathtt{a}), (0, 0, \mathtt{b}), (1, 2, \mathtt{a}), (2, 3, \mathtt{b}), (3, 5, \mathtt{a}), (7, 7, \mathtt{b}), (3, 4, \mathtt{b})$.
  The size of $\SLZ(S)$ is 7.
\end{MyExample}

The self-referencing counterpart is defined as follows:

\begin{Definition}[Self-referencing LZ77 factorization]
  The \emph{self-referencing LZ77 factorization} of string $S$,
  denoted $\SLZ(S)$,
  is a factorization $S = g_1 \cdots g_\mathsf{zs}$ that satisfies
  the following:
  Let $v_i$ denote the beginning position of each factor $g_i$ in
  the factorization $g_1 \cdots g_\mathsf{zs}$,
  that is, $v_i = |g_1 \cdots g_{i-1}|+1$.
  (1) If $i > 1$ and $\max_{1 \leq j < v_i}\{\lcp(S[v_i..], S[j..])\} \geq 1$,
  then for any position $t_i \in \argmax_{1 \leq j < v_i} \lcp(S[v_i..], S[j..])$ in $S$, let $q_i = \lcp(S[v_i..], S[t_i..])$.
  (2) Otherwise, let $q_i = 0$.
  Then, $g_i = S[v_i..v_i+q_i]$ for each $1 \leq i \leq \mathsf{zs}$.
\end{Definition}

Intuitively, each factor $g_i$ of $\SLZ(S)$ is either a fresh letter,
or the shortest prefix of $g_i \cdots g_\zs$
that does not have a previous occurrence beginning in $g_1 \cdots g_{i-1}$.
This means that self-referencing is allowed in $\SLZ(S)$, namely,
the left-most previous occurrence with smallest $t_i$
of each factor $g_i$ may overlap with itself.

The \emph{size} $\zs(S)$ of $\SLZ(S)$ is the number $\zs$ of factors in $\SLZ(S)$.

Likewise, we encode each factor $g_i$ by a triple $(t_i, q_i, \beta_i) \in ([1..n] \times [1..n] \times \Sigma)$,
where $t_i$ is the left-most previous occurrence of $g_i$,
$q_i$ is the length of $g_i$, and $\beta_i$ is the last character of $g_i$.

\begin{MyExample}
  For $S = \mathtt{abaababaabaabaabaabaabb}$,
  $\NLZ(S) = \mathtt{a \mid b \mid aa \mid bab \mid aabaa \mid baabaabaabb \mid}$
  and it can be represented as
  $(0, 0, \mathtt{a}), (0, 0, \mathtt{b}), (1, 2, \mathtt{a}), (2, 3, \mathtt{b}), (3, 5, \mathtt{a}), (7, 11, \mathtt{b})$.
  The size of $\SLZ(S)$ is 6.
\end{MyExample}

\subsection{Communication complexity model}
Our approach is based on
the standard communication complexity model of Yao~\cite{Yao79}
between two parties:
\begin{itemize}
  \item The parties are Alice and Bob;
  \item The problem is a function $f: X \times Y \rightarrow Z$
        for arbitrary sets $X, Y, Z$;
  \item Alice has instance $x \in X$ and Bob has instance $y \in Y$;
  \item The goal of the two parties is to output $f(x, y)$ for a pair $(x, y)$ of instances by a joint computation;
  \item The joint computation (i.e. the communication between Alice and Bob) follows a specified protocol $\mathcal{P}$.
\end{itemize}
The communication complexity~\cite{Yao79} usually refers merely to
the total amount of bits that need to be transferred between Alice and Bob
to compute $f(x, y)$.
In this paper, we follow Bille et al.'s model~\cite{Bille18} where
the communication complexity is evaluated by
a pair $\langle r, b \rangle$ of the number of communication rounds $r$
and the total amount of bits $b$ exchanged in the communication.

In a (Monte-Carlo) randomized public-coin protocol,
each party (Alice/Bob) can access a shared infinitely long
sequence of independent random coin tosses.
The requirement is that
the output has to be correct for every pair of inputs
with probability at least $1 - \epsilon$ for some
$0 < \epsilon < 1/2$,
which is based on the shared random sequence
of coin tosses.
We remark that one can amplify the error rate to an arbitrarily small constant
by paying a constant factor penalty
in the communication complexity.
Note that the public-coin model differs from a randomized private-coin model,
where in the latter the parties do not share a common random sequence
and they can only use their own random sequence.
In a deterministic protocol,
every computation is performed without random sequences.

\subsection{Joint computation of compressed string problems}

In this paper, we also consider
the communication complexity of the Hamming distance problem
between two compressed strings of equal length,
which are compressed by LZ77 without self-references.

\begin{Problem}[Hamming distance with non self-referencing LZ77]\label{prob:HD_NLZ}
  \quad
  \vspace*{-2mm}
  \begin{description}
    \item[Alice's input:] $\NLZ(A)$ for string $A$ of length $n$.
    \item[Bob's input:] $\NLZ(B)$ for string $B$ of length $n$.
    \item[Goal:] Both Alice and Bob obtain the value of $\hd(A, B)$.
  \end{description}
\end{Problem}

The following LCP problem for two strings
compressed by non self-referencing LZ77 
has been considered by Bille et al.~\cite{Bille18}.



\begin{Problem}[LCP with non self-referencing LZ77]\label{prob:LCP_NLZ}
    \quad
    \vspace*{-2mm}
    \begin{description}
      \item[Alice's input:] $\NLZ(A)$ for string $A$.
      \item[Bob's input:] $\NLZ(B)$ for string $B$.
      \item[Goal:] Both Alice and Bob obtain the value of $\lcp(A, B)$.
    \end{description}
\end{Problem}



Bille et al. proposed the following protocol for
a joint computation of the LCP of two strings
compressed by non self-referencing LZ77:

\begin{Theorem}[\cite{Bille18}] \label{theo:LCP_NLZ}
  Suppose that the alphabet $\Sigma$
  and the length $n$ of string $A$ are known to both Alice and Bob.
  Then, there exists a randomized public-coin protocol
  which solves Problem~\ref{prob:LCP_NLZ} with
  communication complexity $\langle O(\log z_\ell), O(\log \ell) \rangle$,
  where $\ell = \lcp(A, B)$ and $z_\ell = \zn(A[1..\ell])$.
\end{Theorem}

The basic idea of Bille et al.'s protocol~\cite{Bille18} is as follows:
In their protocol, the sequences of factors in
the non self-referencing LZ77 factorizations $\NLZ(A)$ and $\NLZ(B)$ are regarded as
strings of respective lengths $\zn(A)$ and $\zn(B)$
over an alphabet $[1..n] \times [1..n] \times \Sigma$.
Then, Alice and Bob jointly compute the LCP of $\NLZ(A)$ and $\NLZ(B)$,
which gives them the first mismatching factors between $\NLZ(A)$ and $\NLZ(B)$.
This LCP of $\NLZ(A)$ and $\NLZ(B)$ is computed by a randomized
protocol for doubling-then-binary searches with
$O(\log z_\ell)$ communication rounds.
Finally, Alice sends the information about her first mismatching factor to Bob, and
he internally computes the LCP of $A$ and $B$.
The total number of bits exchanged is bounded by $O(\log \ell)$.

In Section~\ref{sec:hamming}, we present our protocol for Problem~\ref{prob:HD_NLZ}
of jointly computing the Hamming distance of two strings
compressed by \emph{non self-referencing LZ77}.
The scheme itself is a simple application of the LCP protocol of Theorem~\ref{theo:LCP_NLZ}
for \emph{non self-referencing LZ77},
but our communication complexity analysis is based on
non-trivial combinatorial properties of LZ77 factorization
which, to our knowledge, were not previously known.

\section{Compressed communication complexity of Hamming distance}
\label{sec:hamming}

In this section we show a Monte-Carlo randomized
protocol for
\hbnote*{fixed}{
  Problem~\ref{prob:HD_NLZ}
}
that asks for the Hamming distance $\hd(A, B)$ of strings $A$ and $B$
that are compressed by non self-referencing LZ77.
Our protocol achieves $\langle O(d \log z), O(d \log \ell_{\max}) \rangle$
communication complexity,
where $d = \hd(A, B)$, $z = \zn(A)$, and $\ell_{\max}$ is the largest value
returned by the sub-protocol of the LCP problem for two strings
compressed by non self-referencing LZ77.

The basic idea is to apply the so-called \emph{Kangaroo jumping} method,
namely, if $d$ is the number of mismatching positions
between $A$ and $B$, then one can compute $d = \hd(A, B)$
with at most $d+1$ LCP queries.
More specifically, let $1 \leq i_1 < \cdots < i_d \leq n$
be the sequence of mismatching positions between $A$ and $B$.
By using the protocol of Theorem~\ref{theo:LCP_NLZ} as a black-box,
and also using the fact that $\zn(S) \geq \zn(S[1..j])$
for any prefix $S[1..j]$ of any string $S$,
we immediately obtain the following:

\begin{Lemma} \label{lem:HD_basic_lemma}
  Suppose that the alphabet $\Sigma$
  and the length $n$ of strings $A$ and $B$ are known to both Alice and Bob.
  Then, there exists a randomized public-coin protocol
  which solves Problem~\ref{prob:HD_NLZ} with
  communication complexity
  $\langle O(\sum_{k = 1}^{d} \log \zn(A[i_k+1..])), O(d \log \ell_{\max}) \rangle$,
  where $\ell_{\max} = \max_{1 < k \leq d} \{i_{k} - i_{k-1}+1\}$.
\end{Lemma}

\subsection{On the sizes of non self-referencing LZ77 factorization of suffixes}

Our next question is how large the $\zn(A[i_k+1..])$ term in Lemma~\ref{lem:HD_basic_lemma} can be in comparison to $\zn(A)$.
To answer this question, we consider the following general measure:
For any string of length $n$,
let
\[
  \zeta(n) = \max\{ \zn(S[i..]) / \zn(S) \mid S \in \Sigma^n, 1 < i \leq n\}.
\]

\subsubsection{Lower bound for $\zeta(n)$}

In this subsection, we present a family of strings $S$ such that 
$\zn(S[i..]) > \zn(S)$ for some suffix $S[i..]$,
namely $\zeta(n) > 1$.
More specifically, we show the following:

\begin{Lemma} \label{lem:zeta_lower_bound}
  $\zeta(n)$ is asymptotically lower bounded by $4/3$.
\end{Lemma}

\begin{proof}
  For simplicity, we consider an integer alphabet $\{0, 1, \ldots, \sigma \}$
  of size $\sigma +1$.
  Consider the string
  \[
    S = (0 1 2 \cdots \sigma - 1 \ \ \sigma) (0 1 2 4) (0 1 2 3 4 6) (0 1 2 3 4 5 6 8)  \cdots (0 1 2 \cdots \sigma -2 \ \ \sigma)
  \]
  and its proper suffix
  \[
    S[2..] = (1 2 \cdots \sigma - 1 \ \ \sigma) (0 1 2 4) (0 1 2 3 4 6) (0 1 2 3 4 5 6 8)  \cdots (0 1 2 \cdots \sigma -2 \ \ \sigma).
  \]
  The non self-referencing LZ77 factorization of $S$ and $S[2..]$ are:
  \begin{eqnarray*}
    \NLZ(S) & = & 0 \mid 1 \mid 2 \mid \cdots \mid \sigma - 1 \mid \sigma \mid 0 1 2 4 \hspace*{3mm} \mid 0 1 2 3 4 6 \hspace*{3mm} \mid 0 1 2 3 4 5 6 8 \hspace*{3mm} \mid \cdots \mid 0 1 2 \cdots \sigma -2 \ \ \sigma \hspace*{2mm} \mid \\
    \NLZ(S[2..]) & = & \hspace*{4.5mm} 1 \mid 2 \mid \cdots \mid \sigma - 1 \mid \sigma \mid 0 1 \mid 2 4 \mid 0 1 2 3 \mid 4 6 \mid 0 1 2 3 4 5 \mid 6 8 \mid \cdots \mid 0 1 2 \cdots \mid \sigma -2 \ \ \sigma \mid
  \end{eqnarray*}
  Observe that after the first occurrence of character $\sigma$,
  each factor of $\NLZ(S)$ is divided into two smaller factors in $\NLZ(S[2..])$.
  Since $\zn(S) = |\NLZ(S)| = (\sigma+1) + (\frac{\sigma}{2} - 1) = \frac{3\sigma}{2}$ and
  $\zn(S[2..]) = |\NLZ(S[2..])| = (\sigma) + (\sigma-2) = 2 \sigma - 2$,
  $\zn(S[2..])/\zn(S) = \frac{2 \sigma - 2}{(3\sigma / 2)} = \frac{4}{3} - \frac{2}{3\sigma}$, which tends to $4/3$ as $\sigma$ goes to infinity.
  We finally remark that $|S| = n = \Theta(\sigma^2)$
  which in turn means that $\sigma = \Theta(\sqrt{n})$.
\end{proof}

\begin{MyRemark} \label{rmk:zeta_lower_bound}
  One can generalize the string $S$ of Lemma~\ref{lem:zeta_lower_bound}
  by replacing $0$ with $0^h$ for arbitrarily fixed $1 < h \leq a \cdot \sigma$
  for any constant $a$.
  The upper limit $a \cdot \sigma$ comes from the fact that
  the number of $0$'s in the original string $S$ is exactly
  $\frac{\sigma}{2}$.
  Since $|S| = n = \Theta(\sigma^2)$,
  replacing $0$ by $0^{h}$ with $h < a \cdot \sigma$
  keeps the string length within $O(n)$.
  This implies that one can obtain the asymptotic lower bound
  $4/3$ for any suffix $S[h..]$ of length roughly up to $n - \sqrt{n}$.
\end{MyRemark}

Note also that the factorizations shown in Lemma~\ref{lem:zeta_lower_bound}
coincide with the self-referencing counterparts $\SLZ(S)$ and $\SLZ(S[2..])$,
respectively.
The next corollary immediately follows from Lemma~\ref{lem:zeta_lower_bound}
and Remark~\ref{rmk:zeta_lower_bound}.
\begin{Corollary}
  The Lempel-Ziv 77 factorization with/without self-references is non-monotonic.
\end{Corollary}

\subsubsection{Upper bound for $\zeta(n)$}

Next, we consider an upper bound for $\zeta(n)$.
The tools we use here are the \emph{C-factorization}~\cite{Crochemore84}
without self-references,
and a grammar compression called \emph{AVL-grammar}~\cite{Rytter03}.

\begin{Definition}[Non self-referencing C-factorization]
  The \emph{non self-referencing C-factorization} of string $S$,
  denoted $\CN(S)$,
  is a factorization $S = c_1 \cdots c_\mathsf{cn}$ that satisfies the following:
  Let $w_i$ denote the beginning position of each factor $c_i$ in
  the factorization $c_1 \cdots c_\mathsf{cn}$,
  that is, $w_i = |c_1 \cdots c_{i-1}|+1$.
  (1) If $i > 1$ and $\max_{1 \leq j < w_i}\{\lcp(S[w_i..], S[j..w_i-1])\} \geq 1$,
  then for any position $r_i \in \argmax_{1 \leq j < w_i} \lcp(S[w_i..], S[j..w_i-1])$ in $S$, let $y_i = \lcp(S[w_i..], S[r_i..w_i-1])-1$.
  (2) Otherwise, let $y_i = 0$.
  Then, $c_i = S[w_i..w_i+y_i]$ for each $1 \leq i \leq \mathsf{cn}$.
\end{Definition}
The \emph{size} $\cn(S)$ of $\CN(S)$ is the number $\cn$
of factors in $\CN(S)$.

\begin{MyExample}
  For $S = \mathtt{abaababaabaabaabaabaabb}$,
  $\CN(S) = \mathtt{a \mid b \mid a \mid ab \mid abaab \mid aaba \mid abaabaab \mid b \mid}$
  and its size is $8$.
\end{MyExample}

The difference between $\NLZ(S)$ and $\CN(S)$ is that
while each factor $f_i$ in $\NLZ(S)$ is the shortest prefix of $S[u_i..]$
that does not occur in $S[1..u_i-1]$,
each factor $c_i$ in $\CN(S)$ is the longest prefix of $S[w_i..]$
that occurs in $S[1..w_i-1]$.
This immediately leads to the next lemma.

\begin{Lemma} \label{lem:cn_larger_than_zn}
  For any string $S$, $\cn(S) \geq \zn(S)$.
\end{Lemma}

We also use the next lemma in our upper bound analysis for $\zeta(n)$.

\begin{Lemma} \label{lem:cn_not_larger_than_2zn}
  For any string $S$, $\cn(S) \leq 2\zn(S)$.
\end{Lemma}

\begin{proof}
  \sinote*{removed ``on the contrary''}{%
    Suppose that
  }%
  there are two consecutive factors $c_i, c_{i+1}$ of $\CN(S)$
  and a factor $f_j$ of $\NLZ(S)$
  such that $c_i, c_{i+1}$ are completely contained in $f_i$
  and the ending position of $c_{i+1}$ is less than
  the ending position of $f_j$.
  Since $c_{i}c_{i+1}$ is a substring of $f_j[..|f_j|-1]$
  and $f_j[..|f_j|-1]$ has a previous occurrence in
  $f_1 \cdots f_{j-1}$,
  this contradicts that $c_{i}$ terminated inside $f_j[..|f_j|-1]$.

  Thus the only possible case is that
  $c_{i}c_{i+1}$ occurs as a suffix of $f_j$.
  Note that in this case $c_{i-1}$ cannot occur inside $f_j$
  by the same reasoning as above.
  Therefore, at most two consecutive factors of $\CN(S)$
  can occur completely inside of each factor of $\NLZ(S)$.
  This leads to $\cn(S) \leq 2\zn(S)$.
\end{proof}

An AVL-grammar of a string $S$ is a kind of a
\emph{straight-line program} (\emph{SLP}),
which is a context-free grammar in the Chomsky-normal form which generates only $S$.
The parse-tree of the AVL-grammar is an AVL-tree~\cite{AVL62}
and therefore,
its height is $O(\log n)$ if $n$ is the length of $S$.
Let $\avl(S)$ denote the size (i.e. the number of productions)
in the AVL-grammar for $S$.
Basically, the AVL-grammar for $S$ is constructed from
the C-factorization of $S$,
by introducing at most $O(\log n)$ new productions for each
factor in the C-factorization.
Thus the next lemma holds.

\begin{Lemma}[\cite{Rytter03}] \label{lem:avl_logn_cn}
  For any string $S$ of length $n$,
  $\avl(S) = O(\cn(S) \log n)$.
\end{Lemma}

Now we show our upper bound for $\zeta(n)$.

\begin{Lemma} \label{lem:zeta_upper_bound}
  $\zeta(n) = O(\log n)$.
\end{Lemma}

\begin{proof}
  Suppose we have two AVL-grammars for strings $X$ and $Y$
  of respective sizes $\avl(X)$ and $\avl(Y)$.
  Rytter~\cite{Rytter03} showed how to build
  an AVL-grammar for the concatenated string $XY$
  of size $\avl(X) + \avl(Y) + O(h)$,
  where $h$ is the height of the taller parse tree of the two AVL-grammars
  before the concatenation.
  This procedure is based on a folklore algorithm
  (cf~\cite{Knuth98}) that concatenates two given AVL-trees of
  height $h$ with $O(h)$ node rotations.
  In the concatenation procedure of AVL-grammars,
  $O(1)$ new productions are produced per node rotation.
  Therefore, $O(h)$ new productions are produced in
  the concatenation operation.

  Suppose we have the AVL-grammar of a string $S$ of length $n$.
  It contains $\avl(S)$ productions
  and the height of its parse tree is $h = O(\log n)$
  since an AVL-tree is a balanced binary tree.
  For any proper suffix $S' = S[i..]$ of $S$ with $1 < i \leq n$,
  we split the AVL-grammar into two AVL-grammars,
  one for the prefix $S[1..i]$ and the other for the suffix $S[i..n]$.
  We ignore the former and concentrate on the latter for our analysis.
  Since split operations on a given AVL-grammar can be
  performed in a similar manner to the afore-mentioned concatenation operations,
  we have that $\avl(S') \leq \avl(S) + a \log n$ for some constant $a > 0$.
  Now it follows from Lemma~\ref{lem:cn_larger_than_zn},
  Lemma~\ref{lem:cn_not_larger_than_2zn}, Lemma~\ref{lem:avl_logn_cn},
  and that the size $\cn$ of the C-factorization of any string
  is no more than the number of productions in any SLP generating
  the same string~\cite{Rytter03},
  we have
  \[
    \zn(S') \leq \cn(S') \leq \avl(S') \leq \avl(S) + a \log n \leq a' \cn(S) \log n \leq 2 a' \zn(S) \log n
  \]
  where $a' > 0$ is a constant.
  This gives us $\zn(S') / \zn(S) = O(\log n)$
  for any string $S$ of length $n$ and any of its proper suffix $S'$.
\end{proof}

Since the size $\zn(S)$ of the non self-referencing LZ77 factorization
of any string $S$ of length $n$ is at least $\log n$,
the next corollary is immediate from Lemma~\ref{lem:zeta_upper_bound}:
\begin{Corollary} \label{coro:zeta_upper_bound}
  For any string $S$ and its proper suffix $S'$,
  $\zn(S') / \zn(S) = O(\zn(S))$.
\end{Corollary}

\subsection{Compressed communication complexity of Hamming distance}

Now we have the main result of this section.

\begin{Theorem}
  Suppose that the alphabet $\Sigma$
  and the length $n$ of strings $A$ and $B$ are known to both Alice and Bob.
  Then, there exists a randomized public-coin protocol
  which solves Problem~\ref{prob:HD_NLZ} with
  communication complexity
  $\langle O(d \log \zn), O(d \log \ell_{\max}) \rangle$,
  where $\zn = \zn(A)$ and
  $\ell_{\max} = \max_{1 < k \leq d} \{i_{k} - i_{k-1}+1\}$.
\end{Theorem}

\begin{proof}
  The protocol of Lemma~\ref{lem:HD_basic_lemma}
  has $O(\sum_{k = 1}^{d} \log \zn(A[i_k+1..]))$ rounds.
  By Corollary~\ref{coro:zeta_upper_bound},
  we have that $\zn(A[i_k+1..]) = O(\zn(A)^2)$.
  Therefore,
  $\sum_{k = 1}^{d} \log \zn(A[i_k+1..]) = O(d \log \zn(A))$,
  which proves the theorem.
\end{proof}

\section{Conclusions and open questions}

This paper showed a randomized public-coin protocol
for a joint computation of the Hamming distance 
of two compressed strings.
Our Hamming distance protocol relies on Bille et al.'s
LCP protocol for two strings that are compressed by
non self-referencing LZ77,
while our communication complexity analysis is based on
new combinatorial properties of non self-referencing LZ77 factorization.


As a further research,
it would be interesting to consider 
the communication complexity of the Hamming distance problem
using self-referencing LZ77.
The main question to this regard is whether 
$\zs(S[i..]) = O(\poly(\zs(S)))$ holds for any suffix $S[i..]$ of any string $S$.
In the case of non self-referencing LZ77,
$\zn(S[i..]) = O(\zn(S)^2)$ holds due to Lemma~\ref{coro:zeta_upper_bound}.

\section*{Acknowledgments}
This work was supported by JSPS KAKENHI Grant Numbers
JP18K18002 (YN),
JP20H04141 (HB), JP18H04098 (MT),
and JST PRESTO Grant Number JPMJPR1922 (SI).

\bibliographystyle{abbrv}
\bibliography{ref}

\end{document}